\documentclass[aps,pra,reprint,groupedaddress]{revtex4-1}

\newcommand{\colwiz}[1]{}
\usepackage{setspace}
\usepackage{amsmath,amsthm}

\newtheorem{propn}{Proposition}

\begin{document}

\title{Hardy's Non-locality Paradox and Possibilistic Conditions for Non-locality}

\author{Shane Mansfield}
\email[]{shane.mansfield@comlab.ox.ac.uk}
\affiliation{Oxford University Computing Laboratory, Wolfson Building, Parks Road, Oxford OX1 3QD, U.K.}

\author{Tobias Fritz}
\email[]{tobias.fritz@icfo.es}
\affiliation{ICFO - Institut de Ci\`{e}nces Fot\`{o}niques, Mediterranean Technology Park, 08860 Castelldefels (Barcelona), Spain}

\date{\today}

\begin{abstract}
Hardy's non-locality paradox is a proof without inequalities showing that certain non-local correlations violate local realism. It is `possibilistic' in the sense that one only distinguishes between possible outcomes (positive probability) and impossible outcomes (zero probability). Here we show that Hardy's paradox is quite universal: in any $(2,2,l)$ or $(2,k,2)$ Bell scenario, the occurence of Hardy's paradox is a necessary and sufficient condition for possibilistic non-locality. In particular, it subsumes all ladder paradoxes. This universality of Hardy's paradox is not true more generally: we find a new `proof without inequalities' in the $(2,3,3)$ scenario that can witness non-locality even for correlations that do not display the Hardy paradox. We discuss the ramifications of our results for the computational complexity of recognising possibilistic non-locality.
\end{abstract}

\maketitle

\section{Introduction}

Since the fundamental insight of Bell~\cite{Bell1}\cite{Bell2}, it is known that quantum mechanics cannot be completed to a local realistic theory. This is usually demonstrated by considering spatially separated systems on which certain observables can be measured. Joint measurements across the systems give rise to joint probability distributions for each global choice of observables. Under the assumptions of locality and realism, it can be shown that these correlations need to satisfy certain \emph{Bell inequalities} which can be violated quantum-mechanically, from which Bell's conclusion follows. Inferring that certain correlations are incompatible with local realism is a \emph{non-locality proof}.

A more intuitive approach to non-locality proofs has been pioneered by Greenberger, Horne, Shimony and Zeilinger~\cite{GHZ}\Citep{f20eabee9ae016a}\colwiz{Greenberger1990} as well as Hardy~\Citep{f20eaa114ac38fa}\colwiz{Hardy1993}. This kind of non-locality proof --- we will explain Hardy's `paradox' in section~\ref{hpsect} --- disregards the exact values of the joint outcome probabilities and only records which of them are non-zero and which are zero. In other words, one distinguishes only between \emph{possible outcomes} and \emph{impossible outcomes}, and this turns out to be sufficient for separating quantum mechanics from local realism. Subsequently, several other non-locality `paradoxes' of this type have been found~\cite{f20eb80456562e1}\cite{genHardy1}\cite{genHardy2}.

In this paper, we follow Abramsky~\Citep{f20ea19c93f51ad}\colwiz{Abramsky2010} in considering a general framework for such possibilistic non-locality proofs. More specifically, we study possibilistic Bell inequalities in $(n,k,l)$ Bell scenarios, where $n$ is the number of sites, $k$ is the number of allowed measurements at each site, and $l$ is the number of possible outcomes for each measurement. In this paper, we will only be concerned with scenarios for which $n=2$. What we find is a remarkable universality of Hardy's paradox: it is a necessary and sufficient condition for possibilistic non-locality in all $(2,k,2)$ and $(2,2,l)$ scenarios. However, for the $(2,3,3)$ scenario we find a new possibilistic locality condition which can be violated without the occurence of a Hardy paradox.

\section{Hardy's Non-locality Paradox}\label{hpsect}

The original Hardy paradox concerns the $(2,2,2)$ scenario~\Citep{f20eaa114ac38fa}\colwiz{Hardy1993}. Let us begin by considering this example. To make it more concrete, let's say that the two sites are Alice's lab and Bob's lab, which share a (possibly entangled) quantum state. Each experimenter can choose to make one of two measurements on their subsystem, which we call polarisation and colour. Each measurement has two possible outcomes: $\{ \uparrow, \downarrow\}$ for polarisation, and $\{R,G\}$ for colour. We assume that Alice and Bob perform very many runs of the experiment (each time starting with the same shared state) and then tabulate their results as in table \ref{hptable}. A `$1$' in the table signifies that it was possible to obtain those two outcomes in the same run, and a `$0$' signifies that this never happened. Following~\cite{f20ea19c93f51ad}, we use the term \emph{empirical model} for such a specification of possibilities. In general, any probabilistic empirical model can be transformed into a possibilistic one in a canonical way via \emph{possibilistic collapse}: the process by which all non-zero probabilities are conflated to $1$. 

\begin{table}
\caption{\label{hptable} Empirical model of the Hardy paradox. This is a possibilistic table with `$1$' standing for `possible' and `$0$' standing for `impossible'. The blank entries are unspecified and can be either $0$ or $1$.}
\begin{center}
\begin{singlespace}
\begin{tabular}{cc}
~ & Bob \\
Alice & \begin{tabular}{c|cc|cc|}
~ & $\uparrow$ & $\downarrow$ & $R$ & $G$ \\ \hline
$\uparrow$ & 1 & ~ & ~ & 0 \\
$\downarrow$ & ~ & ~ & ~ & ~ \\ \hline
$R$ & ~ & ~ & 0 & ~ \\
$G$ & 0 & ~ & ~ & ~ \\ \hline
\end{tabular}
\end{tabular}
\end{singlespace}
\end{center}
\end{table}

This partially completed table is Hardy's paradox. The apparent paradox arises because the table tells us that, when both experimenters measured polarisation, it was possible for them to both get the outcome $\uparrow$; but, when one measured polarisation and the other measured colour, it never happened that they could obtain $\uparrow$ and $G$ together. From these statements it seems that whenever $\uparrow$ was measured in one lab, the colour in the other lab must have had the value $R$; and since it was possible for both to get the outcome $\uparrow$, then it should have been possible for both to get the outcome $R$ if the experimenters had instead decided to measure colour on those runs. However, the remaining specified entry in the table tells us that it was not possible for both experimenters to measure $R$. Despite this apparent paradox, such behaviour is actually predicted by quantum mechanics.

Of course, in stating this argument, we have made some implicit assumptions. In particular, we have assumed locality (or, to be more precise, no-signalling) in that we assume that, for each run, Bob's choice of measurement did not affect Alice's outcome and vice versa. Such behaviour could give rise to faster-than-light communication between far distant labs, which is prohibited by special relativity. We have also implicitly assumed realism: that colour and polarisation had definite values even when they were not being measured. A further assumption is that every combined measurement choice has some outcome. This is equivalent to Abramsky's property of measurement locality~\Citep{f20ea19c93f51ad}\colwiz{Abramsky2010}.

We can write the condition for non-occurence of the Hardy paradox in table~\ref{hptable} as a formula in Boolean logic:
\begin{equation*}
p(\uparrow, \uparrow) \quad \to \quad p(\uparrow, G) \lor p(G, \uparrow) \lor p(R,R)\:,
\end{equation*}
where the $p(i,j) \in \{0,1\}$ are the entries of the table, or the possibility values for Alice to obtain outcome $i$ and Bob to obtain outcome $j$. 

For the $(2,2,2)$ scenario, there are $64$ versions of the Hardy paradox which one obtains from table~\ref{hptable} by permuting measurements and/or outcomes. We will continue to use the term paradox throughout this paper, though we draw attention to the fact that this is only an apparent paradox. What the `paradox' states is that models of a certain form cannot satisfy the properties of locality and realism.

\section{Properties of Empirical Models}

In any discussion of locality, realism, etc. it is important to be careful about which properties are being assumed or inferred. A detailed discussion of the properties of possibilistic empirical models is contained in \Citep{f20ea19c93f51ad}\colwiz{Abramsky2010}; and the properties of probabilistic empirical models are discussed in \Citep{f20eaa114ac3b88}\colwiz{Brandenburger2008}. We will now present some of these properties in the context of our tabular representation of $n=2$ empirical models.

We assume from the outset that all the models we deal with satisfy \emph{measurement locality (ML)}: the property that at each site the allowed measurements are independent of which measurements are made at the other sites. For $n=2$, this is equivalent to the property that if the table of a model has any zero box then that box must belong to a row (or column) of zero boxes. This allows us to omit such rows/columns of zero boxes in the tabular representation and to assume that all tables are totally defined on the domain of measurement choices.

\begin{table}
\caption{\label{locreal} Examples of possibilistic empirical models. (a) A deterministic empirical model; (b) a local realistic model; (c) a signalling model.}
\begin{singlespace}
\begin{tabular}{ccc}
\begin{tabular}{c|cc|cc|}
~ & $\uparrow$ & $\downarrow$ & $R$ & $G$ \\ \hline
$\uparrow$ & 1 & 0 & 1 & 0 \\
$\downarrow$ & 0 & 0 & 0 & 0 \\ \hline
$R$ & 1 & 0 & 1 & 0 \\
$G$ & 0 & 0 & 0 & 0 \\ \hline
\end{tabular}
\quad \quad & \quad \quad
\begin{tabular}{c|cc|cc|}
~ & $\uparrow$ & $\downarrow$ & $R$ & $G$ \\ \hline
$\uparrow$ & 1 & 0 & 1 & 0 \\
$\downarrow$ & 1 & 0 & 0 & 1 \\ \hline
$R$ & 1 & 0 & 1 & 0 \\
$G$ & 1 & 0 & 0 & 1 \\ \hline
\end{tabular}
\quad \quad & \quad \quad
\begin{tabular}{c|cc|cc|}
~ & $\uparrow$ & $\downarrow$ & $R$ & $G$ \\ \hline
$\uparrow$ & 1 & 0 & 1 & 0 \\
$\downarrow$ & 0 & 0 & 0 & 0 \\ \hline
$R$ & 0 & 1 & 1 & 0 \\
$G$ & 0 & 0 & 0 & 0 \\ \hline
\end{tabular} \\
(a) \quad \quad & \quad \quad (b) & \quad \quad (c)
\end{tabular}
\end{singlespace}
\end{table}

\emph{(Possibilistic) No-signalling ($NS$)} is the property that the choice of measurement at one site does not affect the possible outcomes at another site. In terms of the tabular representation, this means that if a sub-row has any 1 then that sub-row must have a 1 in each box, and similarly for sub-columns. For example, table~\ref{locreal} (a) and (b) are both no-signalling, while (c) is signalling. In (c), if Alice measures polarisation, then the outcome of a polarisation measurement by Bob has to be $\uparrow$, but if Alice measures colour then Bob always gets $\downarrow$. It can be shown that if an empirical model violates possibilistic no-signalling then it also violates probabilistic no-signalling. The converse does not hold in general~\Citep{f20ea19c93f51ad}\colwiz{Abramsky2010}.

\emph{(Strong) Determinism} is the property that the outcome at each site is uniquely determined by the measurement at that site. In the tabular form, this property says that each box should contain at most one 1, and that the 1s are consistent with no-signalling in that they line up in the same sub-rows/columns where possible. We call such an arrangement of 1s a \emph{deterministic grid}. Table~\ref{locreal} (a) is an example of a deterministic model. By this definition, determinism implies no-signalling.

In order to define local realistic models, we need a notion of stochastic mixtures in our possibilistic setting. We define the mixture of a model $A$ with entries $p^A_{ij}$ and a model $B$ with entries $p^B_{ij}$ to be the model with entries
\begin{equation*}
p_{ij} \equiv p^A_{ij} \lor p^B_{ij}\:,
\end{equation*}
which corresponds to the intuition that an outcome is possible in the mixture if and only if it is possible in at least one component of the mixture.

The \emph{local realistic} models are the models that can be obtained by taking arbitrary mixtures of deterministic models. In the tabular representation, a model is local realistic if and only if every $1$ in its table belongs to some deterministic grid. An example of a local realistic model is table~\ref{locreal} (b). In fact, these are precisely the empirical models that can be described with local hidden variables. The choice of the hidden variable corresponds to saying which deterministic model we're in.

We then obtain the following proposition, which facilitates the application of our results to the usual probabilistic setting:

\begin{propn}
With these definitions, possibilistic collapse takes probabilistic local realistic models to possibilistic local realistic models. Conversely, every possibilistic local realistic model can be written as the possibilistic collapse of a probabilistic one.
\end{propn}

\begin{proof}
The first statement is clear from the fact that a non-trivial convex combination of two probabilities $p^A,p^B\in[0,1]$ is non-zero precisely when at least one of $p^A$ or $p^B$ is non-zero. For the second statement, we simply write a given possibilistic local realistic model as a mixture of deterministic models and assign an arbitrary non-zero probability to each of these models such that the probabilities sum to $1$. This defines a probabilistic local realistic model with the required property.
\end{proof}

We interpret this as saying that a non-locality proof without inequalities exists for given correlations if and only if their possibilistic collapse is not local realistic in our possibilistic sense.

\section{Coarse-Grained Versions of Hardy's Paradox}\label{coarsegrain}

For $(2,2,l)$ scenarios, we consider coarse-grainings of the Hardy paradox. The basic form is the same as in the $(2,2,2)$ case (table~\ref{hptable}), but in the general case (table \ref{genpar}) we have $m_1 \times m_2$, $(l-m_1) \times 1$ and $1 \times (l-m_2)$ subtables of 0s, where $0 < m_2,m_1 < l$. Any empirical model whose table is isomorphic (up to permutations of measurements and outcomes) to table~\ref{genpar} for some particular $m_1$, $m_2$ is said to have a coarse-grained Hardy paradox. We use the notation $H_{(m_1,m_2)}$ for this property. Conditions for the non-occurrence of a paradox can still be written as a logical formula. For table \ref{genpar} the corresponding formula is
\begin{align*}
p(o_1',o_1') & \to \\  \bigvee_{r=m_1+1}^{l} & p(o_r, o_1') \lor \bigvee_{s=m_2+1}^{l} p(o_1',o_s) \lor \bigvee_{\substack{r \in [1,m_1] \\ s \in [1,m_2]}} p(o_r,o_s) \: .
\end{align*}
We use the notation $NH_{(m_1,m_2)}$ for the property that all such formulae are satisfied for a particular model.

\begin{table}
\caption{\label{genpar} A $(2,2,l)$ scenario with a $H_{(m_1,m_2)}$ coarse-grained Hardy paradox.}
\begin{center}
\begin{singlespace}
\begin{equation*}
\begin{array}{c|ccc|cc|}
~ & o_1' & \cdots & o_l' & o_1 \cdots o_{m_2} &  o_{m_2+1} \cdots o_l \\ \hline
o_1' & 1 & ~ & ~ & ~ & 0 \quad \cdots \quad 0 \\
\vdots & ~ & ~ & ~ & ~  & ~ \\
o_l' & ~ & ~ & ~ & ~  & ~ \\ \hline
\begin{array}{c}
o_1 \\ \vdots \\ o_{m_1}
\end{array}
& ~ & ~ & ~ &
\begin{array}{ccc}
0 & \cdots & 0 \\ \vdots & \ddots & \vdots \\ 0 & \cdots & 0
\end{array}
& ~ \\
\begin{array}{c}
o_{m_1+1} \\ \vdots \\ o_l
\end{array}
& 
\begin{array}{c}
0 \\ \vdots \\ 0
\end{array}
& ~ & ~ & ~ & ~ \\ \hline
\end{array}
\end{equation*}
\end{singlespace}
\end{center}
\end{table}

The coarse-graining includes the degenerate values $0$ and $l$ for $m_1$ and $m_2$. The cases $m_1=0$, $m_2=l$ and $m_1=l$, $m_2=0$ are especially interesting. For table~\ref{genpar}, these state that the first sub-column in the lower left box needs to contain some $1$, and, respectively, that the first sub-row in the upper right box needs to contain some $1$. These are the possibilistic no-signalling relations! By permutations of measurements and outcomes, these apply to any $1$ in the table; so for the no-signalling predicate we get
\begin{equation*}
NS = NH_{(0,l)} \land NH_{(l,0)} \: .
\end{equation*}

The case where $m_1=m_2=l$ simply expresses that the lower right box in table \ref{genpar} should contain at least some $1$. This is the normalisation of possibility: at least one outcome has to be possible for each choice of measurements. So given that at least some $1$ occurs somewhere in the table of a no-signalling model, the normalisation of possibility is equivalent to $NH_{(l,l)}$.

These properties and observations extend to all $(2,k,l)$ Bell scenarios by considering $2 \times 2$ subtables.

\section{Universality of Hardy's Paradox}\label{NHsect}

We now write $NH$ for the property that no coarse-grained Hardy paradox occurs.

\begin{propn}\label{equivprop}
For the $(2,2,2)$ scenario, the property of non-occurrence of any coarse-grained paradox is equivalent to possibilistic local realism:
\begin{equation}\label{genequiveq}
\textrm{NH} \quad \leftrightarrow \quad \textrm{(Local Realism)}\:.
\end{equation}
\end{propn}

\begin{proof}
We have already demonstrated in section~\ref{hpsect} that an occurence of the Hardy paradox implies a violation of local realism. It only remains to prove that $NH$ implies local realism. By the observations at the end of the last section, we know in particular that $NH$ implies $NS$, so that we can freely use the latter.

From the earlier definition, a model is local realistic if and only if every $1$ in its tabular representation belongs to some deterministic grid. We begin by choosing an arbitrary $1$ in the table. Without loss of generality (w.l.o.g.) let this be the $1$ in table \ref{equivpf} (a). Then, by $NS$, the first sub-row must have a $1$ in each box, and similarly for the first sub-column. Again w.l.o.g. we let these be the entries in table \ref{equivpf} (b). If the starred entry here is a $1$, this completes the first entry to a deterministic grid and we're done. Assume that the starred entry is a $0$. Then, by no-signalling, we can fill in the $1$s in the lower right box of table~\ref{equivpf} (c). Now, if either of the starred entries in this table is a $1$ then this completes the first entry to a deterministic grid. This must be the case, for if it were not then the $0$s in these places would form a Hardy paradox together with the first entry and the $0$ in the lower right box; but we have assumed the property $NH$.
\end{proof}

\begin{table}
\caption{\label{equivpf} Stages in the proof of proposition~\ref{equivprop}.}
\begin{center}
\begin{singlespace}
\begin{tabular}{ccc}
\begin{tabular}{c|cc|cc|}
~ & ~ & ~ & ~ & ~ \\ \hline
~ & 1 & ~ & ~ & ~ \\
~ & ~ & ~ & ~ & ~ \\ \hline
~ & ~ & ~ & ~ & ~ \\
~ & ~ & ~ & ~ & ~ \\ \hline
\end{tabular}
\quad \quad & \quad \quad
\begin{tabular}{c|cc|cc|}
~ & ~ & ~ & ~ & ~ \\ \hline
~ & 1 & ~ & 1 & ~ \\
~ & ~ & ~ & ~ & ~ \\ \hline
~ & 1 & ~ & * & ~ \\
~ & ~ & ~ & ~ & ~ \\ \hline
\end{tabular}
\quad \quad & \quad \quad
\begin{tabular}{c|cc|cc|}
~ & ~ & ~ & ~ & ~ \\ \hline
~ & 1 & ~ & 1 & * \\
~ & ~ & ~ & ~ & ~ \\ \hline
~ & 1 & ~ & 0 & 1 \\
~ & * & ~ & 1 & ~ \\ \hline
\end{tabular} \\
(a) \quad \quad & \quad \quad (b) & \quad \quad (c)
\end{tabular}
\end{singlespace}
\end{center}
\end{table}

This proposition generalises easily to $(2,2,l)$ scenarios.

\begin{propn}\label{genequivprop}
For $(2,2,l)$ scenarios, the property of non-occurrence of any coarse-grained Hardy paradox is equivalent to local realism; i.e. (\ref{genequiveq}) holds for $(2,2,l)$ scenarios.
\end{propn}

\begin{proof}
Again, it is enough to show that the left-hand side implies the right-hand side while assuming $NS$. If we take an arbitrary $1$ in the table, then w.l.o.g. we can represent the model as in table~\ref{genequivtable} by permuting measurements and outcomes if necessary. Assuming that there is no coarse-grained paradox, at least one of the starred entries must be a 1, and this completes the arbitrarily chosen 1 to a deterministic rectangle.
\end{proof}

\begin{table}
\caption{\label{genequivtable} Taking an arbitrary 1 (upper left) in the table of a no-signalling model forces the table to be of this form.}
\begin{center}
\begin{singlespace}
\begin{tabular}{c|cc|cc|}
~ & ~ & ~ & ~ & ~ \\ \hline
~ & 1 & ~ & $1 \: \cdots \: 1$ & $0 \: \cdots \: 0$ \\
~ & ~ & ~ & ~ & ~ \\ \hline
~ & 
$\begin{array}{c}
1 \\ \vdots \\ 1
\end{array}$
& ~ &
$\begin{array}{ccc}
* & \cdots & * \\
\vdots & \ddots & \vdots \\
* & \cdots & *
\end{array}$
& ~ \\
~ & 
$\begin{array}{c}
0 \\ \vdots \\ 0
\end{array}$
& ~ & ~ & ~
\\ \hline
\end{tabular}
\end{singlespace}
\end{center}
\end{table}

We can also generalise proposition~\ref{equivprop} to  $(2,k,2)$ scenarios.

\begin{propn}\label{2k2}
For $(2,k,2)$ scenarios, the property of non-occurrence of any Hardy paradox is equivalent to local realism; i.e. (\ref{genequiveq}) holds for $(2,k,2)$ scenarios.
\end{propn}

\begin{proof}
By proposition~\ref{equivprop}, we know that this holds for $k=2$, and will show by induction that it holds for all $k$. It is useful to use the tabular representation of models in what follows. In this setting, it has to be shown that every $1$ in a given table can be completed to a deterministic grid of $1$s, assuming that no Hardy paradox occurs. We will show that this property holds for all $k_1 \times k_2$ tables, i.e. for all scenarios with $k_1$ two-outcome measurements for Alice and $k_2$ two-outcome measurements for Bob, given that it holds for all $k_1 \times (k_2 - 1)$ tables and all $(k_1 - 1) \times k_2$ tables.

Firstly, we prove the inductive step in the case where some sub-row or sub-column in the $k_1 \times k_2$ table consists entirely of $0$s. Say we have an outcome sub-column of $0$s for some measurement setting of Bob. Then we pick any $1$ in the table. If this $1$ is in the same measurement setting of Bob as the sub-column of $0$s, then by no-signalling its sub-row has a $1$ in each box of the same setting for Alice. We choose any other of these $1$s, complete it to a deterministic grid in the $k_1 \times (k_2-1)$ table obtained by ignoring the particular setting of Bob. Then, by no-signalling, this must complete to a $k_1 \times k_2$ deterministic grid. If the initial $1$ is in a different measurement setting of Bob to the column of $0$s, one can similarly forget the latter setting and apply the induction assumption to the remaining $k_1 \times (k_2 - 1)$ table. Again, the resulting deterministic grid in the sub-table completes uniquely to the whole table by no-signalling. A similar argument holds for sub-rows of $0$s.

Now we need to prove the inductive step in the case where there are no sub-rows or sub-columns of $0$s. By no-signalling, this is equivalent to no individual box having a sub-row/column of $0$s. Hence we can assume that every box has a diagonal or anti-diagonal of $1$s. We choose an arbitrary $1$ in the table, which w.l.o.g. we can write in the upper left corner. By the inductive hypothesis, this can be completed to a $k_1 \times (k_2-1)$ deterministic grid, which w.l.o.g. we write in the upper left corners of all boxes up to Bob's $(k_2-1)$th setting (see table~\ref{longpftab}).

Assume that this deterministic grid does not complete to Bob's $k_2$th setting. Then there must be a $0$ in the upper right corner of some box(es) of Bob's $k_2$th setting, and a $0$ in the upper left corner of some box(es) in the same setting. In table~\ref{longpftab}, we have illustrated a representative situation, including the diagonals or anti-diagonals that these boxes must have. In order to avoid a Hardy paradox triggered by the $1$s in the top sub-row, we must have $1$s in the starred places, corresponding to all those sub-rows where the $0$ in the $k_2$th setting of Bob occurs on the upper left. But now we can find a deterministic grid including the initial $1$ for table~\ref{longpftab} by choosing the second outcome for Alice in the case of a starred row and the first outcome otherwise, while choosing the first outcome for Bob in all measurements.
\end{proof}

\begin{table}
\caption{\label{longpftab} Table for the proof of proposition~\ref{2k2}.}
\begin{center}
\begin{singlespace}
\begin{tabular}{c|c|c|c|c|}
~ & ~ & ~  & ~ & ~ \\ \hline
~ &
\begin{tabular}{cc}
1 & ~ \\ & ~
\end{tabular}
& $\quad \cdots \quad$ &
\begin{tabular}{cc}
1 & ~ \\ & ~
\end{tabular}
&
\begin{tabular}{cc}
1 & 0 \\ ~ & 1
\end{tabular}
\\ \hline
~ &
\begin{tabular}{cc}
1 & ~ \\ * & ~
\end{tabular}
& $\quad \cdots \quad$ &
\begin{tabular}{cc}
1 & ~ \\ * & ~
\end{tabular}
&
\begin{tabular}{cc}
0 & 1 \\ 1 & ~
\end{tabular}
\\ \hline
~ &
\begin{tabular}{cc}
1 & ~ \\ ~ & ~
\end{tabular}
& $\quad \cdots \quad$ &
\begin{tabular}{cc}
1 & ~ \\ ~ & ~
\end{tabular}
&
~
\\ \hline
~ & ~ & $\ddots$ & ~ & ~ 
\end{tabular}
\end{singlespace}
\end{center}
\end{table}

Hardy's ladder paradox \Citep{f20eb80456562e1}\colwiz{Boschi1997} has been proposed as a generalisation of Hardy's original paradox and was used for experimental tests of quantum non-locality \Citep{f20eb80456562e3}\colwiz{Barbieri2005}. Up to symmetries, there is one ladder paradox for any number of settings $k$. It can be presented neatly in tabular form (table \ref{ladtable}). 

\begin{table}
\caption{\label{ladtable} A ladder paradox. The $(2,2,2)$ ladder paradox is just the standard Hardy paradox.}
\begin{center}
\begin{singlespace}
\begin{tabular}{c|c|c|c|c|}
~ & ~ & ~ & ~ & ~ \\ \hline
~ & \begin{tabular}{cc}
1 & ~ \\ ~ & ~
\end{tabular}
&
\begin{tabular}{cc}
~ & 0 \\ ~ & ~
\end{tabular}
& ~ & ~ \\ \hline
~ &
\begin{tabular}{cc}
~ & ~ \\ 0 & ~
\end{tabular}
&
\begin{tabular}{cc}
* & ~ \\ ~ & ~
\end{tabular}
& $\ddots$ & ~ \\ \hline
~ &
~ & $\ddots$ & ~ &
\begin{tabular}{cc}
~ & 0 \\ ~ & ~
\end{tabular} \\ \hline
~ &
~ & ~ &
\begin{tabular}{cc}
~ & ~ \\ 0 & ~
\end{tabular}
&
\begin{tabular}{cc}
0 & ~ \\ ~ & ~
\end{tabular}
\\ \hline
\end{tabular}
\end{singlespace}
\end{center}
\end{table}

We will not explain here how the ladder paradox is in contradiction with local realism, as our proposition~\ref{2k2} makes it clear that the ladder paradox has to be subsumed by the original Hardy paradox in terms of its strength for proving non-locality. In fact, one can also see this directly: if the starred entry in table~\ref{ladtable} is a $0$, then the Hardy paradox occurs; if it is a $1$, then the ladder paradox for $k-1$ settings gets triggered by this $1$. Applying the argument recursively, we find that either the Hardy paradox occurs somewhere in the table, or the ladder paradox for $2$ settings occurs. Since the latter is again just the original Hardy paradox, we find that Hardy's paradox occurs in any case. Hence the occurence of a ladder paradox always implies the occurence of the original Hardy paradox. 

\begin{propn}\label{ladprop}
For $(2,k,2)$ scenarios, the occurrence of a ladder paradox implies the occurrence of a Hardy paradox.
\end{propn}

\section{Non-universality of Hardy's Paradox}
\label{secnonuniv}

%

The results of the previous section might raise the conjecture that the Hardy paradox could be universal in the same sense for any $(2,k,l)$ scenario. However, we have found that the equivalence of local realism to the absence of Hardy-type non-locality does not hold for $(2,k,l)$ scenarios in general: consider the probabilistic empirical model displayed in table~\ref{counterexa}~(b), for example. It concerns a Bell scenario with three two-outcome measurements for Alice; and one two-outcome and one three-outcome measurement for Bob. (This can easily be expanded to a probabilistic empirical model in the $(2,3,3)$ scenario, but we find the example easier to understand in the form of table~\ref{counterexa}.) By direct inspection, we find that no coarse-grained Hardy paradox occurs for this empirical model. Nevertheless, it displays possibilistic (and hence probabilistic) non-locality: the $1$ in the upper left corner of table~\ref{counterexa}~(a) cannot be completed to a determnistic grid.

In conclusion, the Hardy paradox and its coarse-grainings cannot account for all non-local realistic behaviour in scenarios with at least three settings and at least three outcomes. In general, the non-occurrence of a Hardy paradox is necessary but not sufficient for possibilistic local realism.

\begin{table}
\caption{\label{counterexa} (a) A non-locality proof without inequalities; (b) a probabilistic no-signalling model to which it applies although it displays no (coarse-grained) Hardy paradox.}
\begin{center}
\begin{singlespace}
\begin{tabular}{cc}
\begin{tabular}{c|cc|ccc|}
~ & ~ & ~ & ~ & ~ & ~ \\ \hline
~ & 1 & ~ & 0 & ~ & ~ \\
~ & ~ & ~ & ~ & ~ & ~ \\ \hline
~ & 0 & ~ & ~ & ~ & ~ \\
~ & ~ & ~ & ~ & 0 & ~ \\ \hline
~ & 0 & ~ & ~ & ~ & ~ \\
~ & ~ & ~ & ~ & ~ & 0 \\ \hline
\end{tabular}
\quad \quad & \quad \quad
\begin{tabular}{c|cc|ccc|}
~ & ~ & ~ & ~ & ~ & ~ \\ \hline
~ & $\frac{1}{16}$ & $\frac{3}{16}$ & 0 & $\frac{1}{8}$ & $\frac{1}{8}$ \\
~ & $\frac{3}{16}$ & $\frac{9}{16}$ & $\frac{1}{2}$ & $\frac{1}{8}$ & $\frac{1}{8}$ \\ \hline
~ & 0 & $\frac{1}{2}$ & $\frac{1}{8}$ & $\frac{1}{4}$ & $\frac{1}{8}$ \\
~ & $\frac{1}{4}$ & $\frac{1}{4}$ & $\frac{3}{8}$ & 0 & $\frac{1}{8}$ \\ \hline
~ & 0 & $\frac{1}{2}$ & $\frac{1}{8}$ & $\frac{1}{8}$ & $\frac{1}{4}$ \\
~ & $\frac{1}{4}$ & $\frac{1}{4}$ & $\frac{3}{8}$ & $\frac{1}{8}$ & 0 \\ \hline
\end{tabular}
\\ (a) \quad \quad & \quad \quad (b)
\end{tabular}
\end{singlespace}
\end{center}
\end{table}

\section{Discussion}

We have shown Hardy's non-locality paradox to be relatively universal in terms of non-local or non-realistic behaviour. It is the only non-locality proof without inequalities for $(2,2,l)$ and $(2,k,2)$ Bell scenarios, in the sense that it is a necessary and sufficient condition for possibilistic non-locality. We can even interpret the possibilistic versions of the no-signalling condition and the normalisation of probabilities as degenerate cases of the non-occurence of a coarse-grained Hardy paradox. Moreover, we have found that this universality does not extend to the $(2,3,3)$ possibilistic Bell scenario.

This raises the question of finding other possibilistic non-locality conditions that do not belong to the class of Hardy paradoxes for $k,l \geq 3$. It is pertinent to ask whether such non-locality (for example empirical models of the kind of table~\ref{counterexa}) can be realised in quantum mechanics. It has been demonstrated by one of the authors \Citep{f20eb9de1963362}\colwiz{Fritz2010} that some variants of Hardy-type non-locality that can be realised by Popescu-Rorhlich no-signalling boxes \Citep{f20eabee9adfbf5}\colwiz{Popescu1994} cannot be realised in quantum mechanics. It also remains to be seen how the ideas presented in this paper extend to scenarios with $n>2$, where a two-dimensional tabular representation can no longer be used. The GHZ argument concerns the $(3,2,2)$ scenario \citep{GHZ}; does this rely on an abstraction of Hardy type non-locality?

Finally, the results presented here tell us something about the computational complexity of recognising possiblistic non-locality in certain scenarios, where doing the latter is equivalent to deciding whether a (coarse-grained) Hardy paradox occurs. For $(2,k,2)$ scenarios, this simply amounts to checking all $2\times 2$ sub-tables for such a paradox, which gives an algorithm that is polynomial in the size of the input. For $(2,2,l)$ scenarios, one has to check whether each $1$ in the table can be completed to a deterministic grid; so, following the illustration in table~\ref{genequivtable}, it is to be checked whether there is some $1$ among the starred entries, which is equivalent to the non-occurence of the coarse-grained Hardy paradox. Again, this is clearly polynomial in the size of the input.

So for $(2,2,l)$ and $(2,k,2)$ scenarios, polynomial algorithms can be given. The general case with varying $k$ and $l$ remains an open problem. Our results of section~\ref{secnonuniv} and work by Zavodny \citep{Zavodny2010} lead us to conjecture that this general decidability problem for possibilistic local realist models is NP-hard; as is the case for probabilistic models \Citep{f20eb80456569f9}\colwiz{Pitowsky1991}. This gives reason to suspect that it may not be possible to obtain a classification of conditions that are necessary and sufficient for possibilistic local realism in full generality.

\begin{acknowledgements}
SM would like to thank Samson Abramsky and Rui Barbosa for valuable discussions. TF is supported by the EU STREP QCS.
\end{acknowledgements}

\bibliographystyle{plain}
\bibliography{hardybib}

\end{document}